\documentclass[11pt]{article}
\usepackage[T1]{fontenc}

\usepackage{theorem}
\usepackage{subfigure}
\usepackage{epsfig}
\usepackage{amssymb}
\usepackage{fullpage}
\usepackage{amsmath}
\usepackage{ifthen}
\usepackage{verbatim}
\usepackage{hyperref}
\usepackage{xspace}
\usepackage{bm}
\usepackage{setspace}
\usepackage{tablefootnote}
\usepackage{xcolor}
\usepackage{xurl}
\newtheorem{theorem}{Theorem}[section]
\newtheorem{lemma}[theorem]{Lemma}

\newtheorem{corollary}	[theorem]	{Corollary}
\newtheorem{prop}	[theorem]	{Proposition}

{\theorembodyfont{\rmfamily} \newtheorem{definition}
[theorem]	{Definition}}
{\theorembodyfont{\rmfamily} }
{\theorembodyfont{\rmfamily} \newtheorem{example}		[theorem]
{Example}}
{\theorembodyfont{\rmfamily} }
{\theorembodyfont{\rmfamily} }
{\theorembodyfont{\rmfamily} }
{\theorembodyfont{\rmfamily} }
{\theorembodyfont{\rmfamily} }
{\theorembodyfont{\rmfamily} }
{\theorembodyfont{\rmfamily} }
\theoremstyle{break}
{\theorembodyfont{\rmfamily} }

\newenvironment{proof}{\noindent {\em {Proof:}}}{$\blacksquare$\vskip
\belowdisplayskip}

\newcommand{\prob}[2][]{\text{\bf Pr}\ifthenelse{\not\equal{}{#1}}{_{#1}}{}\!\left[#2\right]}
\newcommand{\expect}[2][]{\text{\bf E}\ifthenelse{\not\equal{}{#1}}{_{#1}}{}\!\left[#2\right]}

\def\sm{\setminus}
\def\sse{\subseteq}
\def\eps{\epsilon}

\newcommand{\tfm}{(\incs,\conf,\prices,\revs)}

\newcommand{\gs}{(I,J,\txss)}
\newcommand{\I}{\mathcal{I}}

\newcommand{\util}{u}

\newcommand{\utili}[1][i]{{\util_{#1}}}

\newcommand{\bid}{b}
\newcommand{\bids}{{\mathbf \bid}}

\newcommand{\bidsmi}{{\mathbf \bid}_{-i}}

\newcommand{\bidi}[1][i]{{\bid_{#1}}}

\newcommand{\dist}{D}
\newcommand{\dists}{{\mathbf \dist}}

\newcommand{\block}{B}
\newcommand{\blocks}{{\mathbf \block}}

\newcommand{\blocksmj}{{\mathbf \block}_{-j}}
\newcommand{\blockj}[1][j]{{\block_{#1}}}

\newcommand{\superblock}{\mathcal{B}}
\newcommand{\conf}{\mathcal{C}}

\newcommand{\txs}{S}
\newcommand{\txss}{{\mathbf \txs}}

\newcommand{\txsj}[1][j]{{\txs_{#1}}}

\newcommand{\rev}{\pi}
\newcommand{\revs}{{\mathbf \rev}}

\newcommand{\revj}[1][j]{{\rev_{#1}}}

\newcommand{\val}{v}
\newcommand{\vals}{{\mathbf \val}}

\newcommand{\vali}[1][i]{{\val_{#1}}}

\newcommand{\alloc}{x}
\newcommand{\allocs}{{\mathbf \alloc}}

\newcommand{\alloci}[1][i]{\alloc_{#1}}

\newcommand{\inc}{y}
\newcommand{\incs}{{\mathbf \inc}}

\newcommand{\price}{p}
\newcommand{\prices}{{\mathbf \price}}

\newcommand{\pricei}[1][i]{\price_{#1}}

\title{Transaction Fee Mechanism Design for Leaderless Blockchain Protocols}

\author{Pranav Garimidi\thanks{a16z
    crypto. Email: \texttt{pgarimidi@a16z.com}.} \and 
Lioba Heimbach\thanks{ETH Zurich. Email:
    \texttt{hlioba@ethz.ch}. Work performed in part during an
    internship at a16z crypto.} 
\and Tim
  Roughgarden\thanks{Columbia University \& a16z crypto. 
Author's research at Columbia
  University supported in part by NSF awards CCF-2006737
and CNS-2212745.
Email: \texttt{tim.roughgarden@gmail.com}. 
}}

\begin{document}

\maketitle

\begin{abstract}
We initiate the study of transaction fee mechanism
design for blockchain protocols in which multiple block producers
contribute to the production of each block. Our contributions include:
\begin{itemize}

\item We propose an extensive-form (multi-stage) game model to reason
  about the game theory of multi-proposer transaction
  fee mechanisms. 

\item We define the {\em strongly BPIC} property to capture
the idea that all block producers should be motivated to behave as
intended: for every user bid profile, 
following the
  intended allocation rule is a Nash equilibrium for block producers
  that Pareto dominates all other Nash equilibria.

\item We propose the {\em first-price auction with equal sharing
    (FPA-EQ)} mechanism as an attractive solution to the
  multi-proposer transaction fee mechanism design problem. We prove
  that the mechanism is strongly BPIC and guarantees at least a 63.2\%
  fraction of the maximum-possible expected welfare at equilibrium. 

\item We prove that the compromises made by the FPA-EQ mechanism are
  qualitatively necessary: no strongly BPIC mechanism with non-trivial
  welfare guarantees can be DSIC, and no strongly BPIC mechanism can
  guarantee optimal welfare at equilibrium.

\end{itemize}
\end{abstract}

\section{Introduction}\label{s:intro}

\subsection{Transaction Fee Mechanisms}\label{ss:tfm}

A {\em transaction fee mechanism} is the component of a blockchain
protocol responsible for deciding which pending transactions should be
included for processing, and what the creators of those transactions
should pay for the privilege of execution in the blockchain's virtual
machine. For example, the Bitcoin protocol~\cite{bitcoin} launched
with a first-price auction as its transaction fee mechanism (which
remains in use to this day): users submit bids along with their
transactions; should a transaction be included in a block, its bid is
transferred from the user to the producer of that block. Block
producers are then expected to assemble blocks that maximize their
revenue (i.e., the sum of the bids of the included transactions) subject to
a block size constraint. The Ethereum protocol also launched with a
first-price auction as its transaction fee
mechanism~\cite{eth_yellow_paper} but, in order to achieve stronger
incentive-compatibility guarantees, the protocol's first-price auction
was swapped out in August~2021 in favor of a more sophisticated
transaction fee mechanism known as EIP-1559~\cite{eip1559spec}.  Since
the initial economic analysis of EIP-1559~\cite{eip1559}, a large body of
research has been developed to explore the design space of transaction
fee mechanisms and to assess different designs through the lenses of
incentive-compatibility (both for users and for block producers),
collusion-resistance, welfare, revenue, and more; see
Section~\ref{ss:rw} for an overview.

The entire literature on transaction fee mechanisms considers only
{\em leader-based} blockchain protocols in which each block is
assembled unilaterally by a single block producer (like a Bitcoin
miner or an Ethereum validator) with monopoly power over the contents
of its block. This focus reflects the fact that the vast majority of
the major blockchain protocols deployed to-date are leader-based in
this sense. For example, all longest-chain protocols in the spirit of
Bitcoin and PBFT-type protocols in the vein of
Tendermint~\cite{tendermint} are leader-based. 
But the state-of-the-art in consensus protocol design is evolving, and
the design of transaction fee mechanisms must evolve with them.

\subsection{Leaderless Blockchain Protocols}\label{ss:leaderless}

A new generation of consensus protocols, known as {\em DAG-based
  consensus}, is exploring {\em leaderless} protocol designs (where
``DAG'' stands for ``directed acyclic graph'').
In DAG-based consensus protocols,
multiple
validators build and propose blocks concurrently.
Together, the validators build a DAG: whenever a block is
proposed by a validator, the block references blocks from previous
rounds, effectively voting on these referenced blocks. In each round,
some of the blocks (sometimes referred to as
anchor blocks) %
are used as checkpoints in the DAG structure for
consensus. When 
an anchor block
is finalized, transactions from all
blocks in its causal history that have not been executed previously
are deterministically ordered and staged for execution.

Recently, DAG-based consensus protocols have experienced a rise in the
blockchain ecosystem, with Sui running Mysticeti in
production~\cite{sui_consensus} and other projects such as Aptos
planning to transition to DAG-based consensus. The main reason for the
rise in popularity of DAG-based consensus protocols is the significant
throughput improvements they achieve in comparison to single-leader
BFT consensus
protocols~\cite{baird2016swirlds,keidar2021all,danezis2022narwhal,spiegelman2022bullshark,spiegelman2022bullsharkpartial}. These
throughput improvements stem primarily from two design choices: (1)
the separation of the communication and consensus layers, and (2) the
use of simultaneous block proposals by all validators to overcome the
bottlenecks that arise with the single-leader approach (in effect,
spreading what had been a concentrated workload for the leader across
all validators). Further, while
DAG-based protocols initially suffered from increased latency, current
protocols achieve almost optimal latency (up to one extra communication
round)~\cite{babel2023mysticeti,arun2024shoal++}. Finally, DAG-based
protocols have the advantage that they generally recover quickly from
crash failures of leaders given that they have backup leaders in
place~\cite{babel2023mysticeti}.

\subsection{Our Contributions}\label{ss:contrib}

This paper initiates the study of transaction fee mechanism
design for blockchain protocols in which multiple block producers
contribute to the production of each block. To reason about such
mechanisms, several new modeling and design challenges must be
addressed:
\begin{itemize}

\item Transaction fee mechanism design with a single block producer
  can focus on equilibria purely from the perspective of users, with
  the block producer best responding to the resulting bids;
  with multiple block producers, the ``game within the game,'' meaning the interaction between the
incentives of different block producers, must be explicitly modeled and analyzed.

\item The design of a transaction fee mechanism must now specify how
  proposals from multiple block proposers are aggregated into a single
  block of confirmed transactions.

\item The design of a transaction fee mechanism must now specify how
any unburned fee revenue from users is distributed between the
different block proposers.

\end{itemize}
This paper offers the following contributions:
\begin{itemize}

\item We formally model the game theory of multi-proposer transaction
  fee mechanisms via extensive-form (multi-stage) games.  We define
  incentive-compatibility for block producers in a multi-proposer
  transaction fee mechanism, focusing on a condition we call {\em
    strongly BPIC}. Intuitively, a transaction fee mechanism is
  strongly BPIC if, no matter what the user bids, following the
  intended allocation rule is a Nash equilibrium for block producers
  that Pareto dominates all other Nash equilibria.
 While we use our model specifically to study the welfare
  guarantees achievable by multi-proposer transaction fee mechanisms,
  it should serve as the appropriate starting point to study a number
  of other potential benefits of such mechanisms (see
  Section~\ref{ss:rw} below for examples).

\item We propose the {\em first-price auction with equal sharing
    (FPA-EQ)} mechanism as an attractive solution to the
  multi-proposer transaction fee mechanism design problem.\footnote{In Sui's current transaction fee mechanism, users pay
    their bids, and fee revenue is shared
    with validators pro rata, proportional to validator stake
    weights (Alberto Sonnino, personal communication, October
    2024). The FPA-EQ mechanism can be viewed as the refinement of
    this mechanism in which transaction fees are shared with
    validators in proportion to the number of blocks they have
    contributed to.}
We prove
  that the mechanism is strongly BPIC and guarantees near-optimal
  welfare.
Precisely, for every joint distribution over (possibly
  correlated) user valuations, for every subgame perfect equilibrium
  of the mechanism in which block producers play only Pareto-dominant Nash
  equilibria, the expected equilibrium welfare is at least
  $1-\tfrac{1}{e} \approx 63.2\%$ of the maximum possible. 
Our analysis here brings, for the first time, the powerful toolbox on
``price of anarchy'' bounds to bear on the analysis of transaction fee
mechanisms.
A simple example shows that the bound of 63.2\% is tight in the worst
  case.

\item We prove that the compromises made by the FPA-EQ mechanism are
  qualitatively necessary: no strongly BPIC mechanism with non-trivial
  welfare guarantees can be DSIC (i.e., with truthful bidding a
  dominant strategy for users), and no strongly BPIC mechanism can
  guarantee optimal welfare at equilibrium.

\end{itemize}

\subsection{Related Work}\label{ss:rw}

\paragraph{General TFM literature.}
There is a long line of work studying transaction fee mechanisms for
single-leader protocols, particularly focusing on Ethereum and
Bitcoin. Our model of transaction fee mechanism design closely follows
the line of work initiated by
Roughgarden~\cite{roughgarden2021transaction} to analyze the EIP-1559
mechanism~\cite{eip1559spec}. Before this, research on Bitcoin's fee
market focused on monopolistic pricing
mechanisms~\cite{LSZ19,Y18}. More recent work in this area
includes~\cite{nisan2023serial} and~\cite{gafni2022greedy}. Building
off Roughgarden’s model, Chung and Shi~\cite{chung2023foundations}
show that achieving an ideal TFM is impossible. They attempt to
address these impossibilities using
cryptography~\cite{shi2022can,wu2023maximizing}, but even with
cryptographic methods, perfect TFMs remain unachievable. Furthermore,
Chung et al.~\cite{chung2024collusion} and Gafni and
Yaish~\cite{gafni2024barriers} show that no mechanism can be incentive
compatible for the users and the block producers while also being
collusion-resistant. All of these impossibility results carry over to
our context, as a single-leader protocol is a special case of a
multiple-leader protocol. Although the majority of this work is
prior-free, Zhao et al.~\cite{zhao2022bayesian} consider a Bayesian
setup, demonstrating ways to circumvent these impossibility results in
cases where bidders have i.i.d.\ valuations. Other works explore TFM
dynamics over multiple blocks~\cite{ferreira2021dynamic,monnot_aft}
and incorporate maximal extractable value (MEV) into traditional TFM
models~\cite{eobp1}.

\paragraph{DAG-based consensus.}
Hashgraph~\cite{baird2016swirlds} was the first protocol to introduce
a DAG-based consensus protocol. It separated the communication layer
and the consensus logic, with the communication layer constructing a
DAG of messages which is then used by the consensus protocol. Later
protocols adopted a round-based structure within the DAG to design
more efficient asynchronous BFT
protocols~\cite{gkagol2019aleph,keidar2021all,danezis2022narwhal,spiegelman2022bullshark}. Among
these, Bullshark's partially synchronous variant became the first
widely deployed DAG-based consensus protocol, notably used in the Sui
blockchain~\cite{spiegelman2022bullsharkpartial}. Since Bullshark's
deployment, a number of papers have focused on reducing the latency
of DAG-based consensus
protocols~\cite{keidar2022cordial,malkhi2023bbca,spiegelman2023shoal,babel2023mysticeti,shrestha2024sailfish,arun2024shoal++}. The
designs in these
papers generally either move towards uncertified DAGs (which do not
require explicit certification) or interleave multiple instances of
the Bullshark protocol on a shared
DAG. Mysticeti~\cite{babel2023mysticeti} has replaced Bullshark as the
consensus protocol used by the Sui blockchain~\cite{sui_consensus}.

\paragraph{Economics of multiple leaders.}
There is a modest amount of work concerning 
the incentives faced by validators in multi-leader
protocols. %
Zhang and
Kate~\cite{zhang2024} show how DAG-based consensus protocols can be manipulated
for MEV, while Malkhi et al.~\cite{malkhi2023maximal} propose MEV
protection for such protocols. Fox et al.~\cite{FoxPR23} look at the cost of
censorship in single-leader protocols and show how TFMs specific to
multi-leader protocols could potentially be used to significantly
increase the cost of censorship. The Solana community has been
considering whether to introduce multiple leaders to promote
competition between block producers for the benefit of
users~\cite{toly_mcl}, and Ethereum is planning on incorporating some
of the ideas from multi-proposer architectures through
FOCIL~\cite{focil} to increase Ethereum's censorship resistance. 
The present work
does not directly address questions around MEV, censorship, or
explicit competition between block producers, 
but we believe that the model that we introduce in the next section can
serve as the starting point for a formal study of these questions.

\section{The Model}\label{s:model}

This section defines our game-theoretic model, the design space of
transaction fee mechanisms, several notions of
incentive-compatibility, and approximate welfare guarantees.

\subsection{The Players}\label{ss:players}

Games have three ingredients: players, strategy spaces, and payoffs.
For transaction fee mechanisms (TFMs), there are two types of
self-interested players, users and block producers (BPs). We discuss
each in turn.

We assume that the set~$I=\{1,2,\ldots,n\}$ of users is known, and
that each is identified with a single transaction; we refer to users
and transactions interchangeably.  We assume that user~$i$ has a
private valuation~$\vali$ for the inclusion of its transaction in the
next block,
and that transaction validity does not depend on transaction
ordering.
When discussing Bayes-Nash equilibria (as is necessary
when discussing TFMs without dominant strategies, such as
variants of first-price auctions), we assume that user
valuations~$\vals$ are drawn from a prior distribution~$\dists$ that
is common knowledge among the users.\footnote{We allow valuation
  distributions to have atoms at zero (or at other values), in which case
  the number of (non-null) players can be thought of as stochastic rather
  than known.}
User valuations may be correlated; that is, $\dists$ need not be a
product distribution.

We consider TFMs in which each user attaches a nonnegative
bid~$\bidi$ to its transaction (thus, the strategy space of user~$i$
is the possible choices of~$\bidi$).  We assume that each user has a
quasi-linear utility function, meaning that its payoff is the value it
receives ($\vali$ if its transaction is included in the next block and
0 otherwise) minus the payment it makes. (Utilities functions will be
stated more formally following the definition of TFMs; see
Section~\ref{ss:ic}.)

We also consider a set~$J=\{1,2,\ldots,m\}$ of BPs.
BP strategies correspond to blocks, where for a known block size~$k$,
a {\em block} is a set of at most~$k$ transactions (together with the
bids of those transactions). We assume that each BP $j \in J$ has an
associated subset~$\txsj$ of transactions that it can include in its
block; we refer to the special case in which $\txsj = I$ for all~$j
\in J$ as the {\em BP-symmetric setting} and the case of general
$\txsj$'s as the {\em BP-asymmetric setting}. A block is {\em
  feasible} for BP~$j$ if it includes only transactions of~$\txsj$
and, possibly, additional transactions created by~$j$ itself (e.g., in
order to manipulate a TFM's payment rule). 
We assume that the $\txsj$'s are common knowledge.
The payoff of a BP is defined as the revenue it earns 
from transactions other than its own minus any payments it makes for shill transactions.

\subsection{The Game}\label{ss:game}

TFM outcomes are, intuitively, determined by a two-stage process:
users decide which bids to attach to their transactions, and BPs then
decide which transactions to include. Previous work on TFMs, with a
single BP, could essentially model the process with one stage (with
the understanding that the BP will respond to users' bids with its
favorite block). With multiple BPs best responding to each other (in
addition to users' bids), it is important to explicitly model the
block formation process as a two-stage game. We do this next, using the
standard formalism for extensive-form games
(e.g.~\cite{fudenberg_tirole}).

\vspace{-.5\baselineskip}

\paragraph{Game trees.}
To review, an extensive form game is defined by a rooted tree (the
{\em game tree}). Each node represents a single action to be taken by
a single player, with the node labeled with that player and edges
leading to the node's children labeled with the possible actions.
Each leaf of the tree corresponds to an outcome of the game, and is
labeled with players' payoffs in that outcome. Thus, root-leaf paths
of the game tree correspond to action sequences that terminate in an
outcome of the game.  It is a convenient tradition to allow nodes that
are labeled with a non-strategic ``Nature'' player, indicating that
the action at that node is chosen at random from a distribution that
is common knowledge.
Finally, for each player, the nodes labeled with that player are
partitioned into information sets. An information set represents a set
of nodes that are indistinguishable to the player at the time it must
take an action (and thus, the same action must be taken by a player at
all nodes in the same information set).

To model behavior in TFMs with multiple BPs, we consider a game tree
with $n+m+1$ levels. (The outcomes and payoffs at the leaves of this
tree will depend on the choice of the TFM, but the tree structure is
independent of the particular TFM.)
At level~0, Nature moves and chooses valuations~$\vals$ for all users
from the assumed prior $\dists$.
At each level~$l=1,2,\ldots,n$, user~$l$ selects its bid
$\bid_l$. 
Information sets are defined for user~$l$ so that its choice
of bid depends only on its own valuation $\val_l$ (and not on the
other valuations $\vals_{-l}$ determined at level~0 or the bids chosen
by users~$i \in \{1,2,\ldots,l-1\}$ at earlier levels).
At each level~$l=n+1,n+2,\ldots,n+m$, BP~$j=l-n$ selects its block~$\blockj$.
Information sets are defined for a BP so that its choice
of block can depend on the bids~$\bids$ chosen by users but not on
the blocks chosen by the other BPs. 
\footnote{Thus, BPs engage in a complete-information game, with
    the full bid vector $\bids$ and the $\txsj$'s known to all
    BPs. A good (though possibly difficult) direction for future work
    is to consider an incomplete-information generalization of our model.
With our assumptions, users can effectively treat BPs as carrying out
the welfare-maximizing allocation rule. In an incomplete-information
setup, users would effectively be submitting bids to a randomized
allocation rule induced by some (perhaps impossible-to-characterize)
Bayes-Nash equilibrium played by the BPs.}

\vspace{-.5\baselineskip}

\paragraph{Subgame perfect equilibria.}
Our analysis uses what is arguably the most canonical equilibrium
concept in extensive-form games, namely subgame perfect equilibria.
In such a game, a strategy for a player is defined by a
mapping from each of its information sets to one of the actions
available at that information set. In our model of TFMs, a user has
one information set for each realization of its valuation, and a BP
has one information set for each user bid vector. Thus, a user
strategy is simply a bidding strategy, meaning a mapping
$\vali \mapsto \bidi$ from valuations to bids. A BP strategy is a
mapping $\bids \mapsto \blockj$ from user bid vectors to feasible
blocks. Thus, leaves of the game tree are effectively labeled by
$\vals$ (Nature's action at level~0), $\bids$ (users' actions at
levels~1 through~$n$), and $\blocks$ (BPs' actions at levels~$n+1$
through~$n+m$); these, in conjunction with the choice of a TFM, will
define the player payoffs at this outcome.

A strategy profile in an extensive-form game is called a Nash
equilibrium if the usual best-response condition holds: no player can
strictly improve its expected payoff through a unilateral
deviation to a different mapping of its information sets to
actions. That is, each player is best responding to the strategies
chosen by the other players.

Every node of a game tree induces a rooted subtree that can be
regarded as an extensive-form game in its own right. Similarly, every
strategy of an extensive-form game induces a strategy in each of its
subgames. A strategy profile of an extensive-form game is called a
{\em subgame perfect equilibrium (SPE)} if, for each if its subgames,
the induced strategy profile is a Nash equilibrium. Intuitively, even
after ``fast forwarding'' to an arbitrary node of the game tree, play
from then on constitutes a Nash equilibrium.\footnote{Without the
  subgame perfect refinement, Nash equilibria of extensive-form games
  allow players to play arbitrary strategies in subgames that are
  reached with probability~0.}

Intuitively, in our model of TFMs with multiple BPs, the SPE condition
translates to (i) users play a Bayes-Nash equilibrium relative to the
BP equilibrium strategies; (ii) BPs play a Nash equilibrium relative
to the user bids.\footnote{We do not model how BPs coordinate
    on a given equilibrium.
Microfounding the assumption that BPs reach an equilibrium (e.g.,
through experience from repeated play, explicit coordination based on
transaction hashes, or other means) is an
interesting direction for future research.}

\subsection{Transaction Fee Mechanisms}\label{ss:tfm_def}

A TFM is specified by four ingredients: an inclusion rule (the blocks of
transactions that the BPs are expected to contribute), a confirmation
rule (given the proposed blocks, which transactions are confirmed for
execution), a payment rule (given the proposed blocks, what
the creators of confirmed transactions pay), and a distribution rule
(given the proposed blocks,
the revenue received by BPs). Because BPs have unilateral
control over the transactions they include, the inclusion rule can
only be viewed as a recommendation to BPs; the other three rules are
hard-coded into the code of a blockchain protocol and cannot be
manipulated by BPs.

We next define these four ingredients formally, along with a number of
examples that illustrate the definitions and demonstrate the richness
of the TFM design space with multiple BPs. These rules are all defined
with respect to a commonly known {\em game structure}, meaning a
player set~$I$, a BP set~$J$, BP transaction sets
$\txs_1,\ldots,\txs_m$, and a block size~$k$.\footnote{The valuation
  distribution $\dists$ is not part of the game structure; in this
  sense, a TFM is by definition prior-free.}
Recall that a block~$\blockj$ is {\em feasible
  for~$j$} if it includes only transactions of~$S_j$ and, possibly,
transactions that~$j$ itself created (along with the bids attached to
the included transactions).
When we are concerned only with the transactions included in a block
and not the attached bids, we sometimes abuse notation and treat a
block as a subset of $I$.
We call a profile $\blocks=(\block_1,\ldots,\block_m)$ of
block choices an {\em allocation}, and call an allocation
{\em feasible} if each of its blocks~$\block_j$ is feasible for the
corresponding BP~$j$. We call an allocation {\em shill-free} if, for
each of its blocks, only user-submitted transactions are included
(i.e., $\blockj \sse \txsj$ for every BP~$j$).
Note that the same transaction may be included in more
than one block of an allocation.
We denote by $T(\blocks) = \cup_{j\in J} \blockj$ the transactions that are
included (at least once) in an allocation~$\blocks$.

\vspace{-.5\baselineskip}
\paragraph{Inclusion rules.}
An inclusion rule can be thought of as a recommendation of the
strategies that BPs should play in each information set of the
extensive-form game described in Section~\ref{ss:game}. 
Formally, with respect to a game structure, an {\em inclusion rule} is
a function~$\incs:\bids \mapsto \blocks$ mapping user bids vectors to
feasible allocations.

\begin{example}[Welfare-Maximizing (WM) Inclusion Rule]
    This inclusion rule maps
each bid vector to a feasible shill-free allocation that maximizes the
sum of the bids of the included transactions (breaking ties using some
consistent rule). For TFMs with first-price payment rules (see below),
this inclusion rule can be interpreted as maximizing the total fees paid by users.
\end{example}

\vspace{-.5\baselineskip}
\paragraph{Confirmation rules.}
A confirmation rule specifies which of the included transactions are
confirmed for execution. Formally, with respect to a game structure, a
{\em confirmation rule} is a
function~$\conf: \blocks \mapsto \superblock$ that maps each feasible
allocation $\blocks$ to a set
$\superblock \sse T(\blocks)$ of confirmed transactions.  Note
that while a transaction may be included in multiple blocks, it can
only be confirmed once. 

\begin{example}[First-Price Auction (FPA) Confirmation Rule]
    This confirmation rule confirms every transaction that is included at least once:
$\conf(\blocks) = T(\blocks)$.\footnote{One reason to include
  unconfirmed transactions is to use their bids to set prices for the
  confirmed transactions, in the spirit of a second-price auction. For
  more details, see the {\em second-price auction (SPA)} confirmation
  rule described in Section~\ref{app:tfm_ex}.}
\end{example}

\vspace{-.5\baselineskip}
\paragraph{Payment rules.}
A payment rule specifies the transaction fee paid by the creator of an
included transaction.
Formally, with respect to a game structure, a
{\em payment rule} is a
function~$\prices$ that maps each feasible
allocation $\blocks$ to a set of~$n$ nonnegative
numbers (one per user).

\begin{example}[First-Price Auction (FPA) Payment Rule]
    This payment rule charges the creator of an included transaction its bid:
$\pricei(\blocks) = \bidi$ if $i \in T(\blocks)$ and
$\pricei(\blocks) = 0$ otherwise.
\end{example}

\vspace{-.5\baselineskip}
\paragraph{Distribution rules.}
A distribution rule specifies the revenue earned by each BP from the set
of included transactions.
Formally, with respect to a game structure, a
{\em distribution rule} is a
function~$\revs$ that maps each feasible
allocation $\blocks$ to a set of~$m$ nonnegative
numbers (one per BP).

\begin{example}[Equal-Share (EQ) Distribution Rule]
    This distribution rule splits the bid of each included transaction equally between the BPs:
for all~$j$,
\begin{equation}\label{eq:eq}
\revj(\blocks) = \frac{1}{m} \sum_{i \in T(\blocks)} \bidi.
\end{equation}
\end{example}

\vspace{-.5\baselineskip}
\paragraph{TFMs.} A {\em transaction fee mechanism (TFM)} is then a
tuple $\tfm$. We restrict attention to TFMs that satisfy the following
properties (which are also shared by all TFMs that have been deployed
in practice to-date): (i) deterministic, meaning that $\incs$, $\conf$,
$\prices$, and $\revs$ are all deterministic functions of their inputs;
and (ii) ex post individually rational, meaning that $\pricei(\blocks) =
0$ if user~$i$'s transaction is not confirmed by the TFM (i.e., $i
\not\in \conf(\blocks)$) and $\pricei(\blocks) \le \bidi$ otherwise;
(iii) weakly budget-balanced, meaning that users' payments always
cover BP revenue: $\sum_{j \in J} \revj(\blocks) \le \sum_{i \in I}
\pricei(\blocks)$ for every feasible allocation
$\blocks$.\footnote{As an extension to~(iii), money-printing in the
  form of inflationary rewards
  (like a block reward) can be added to a TFM without affecting its
  incentive or welfare properties, provided the rewards are
  the same no matter which feasible allocation $\blocks$ is chosen by
  the BPs.}
We do allow the user payments to exceed the BP revenue, in which we
case the remaining user payments are burned (or otherwise redirected
away from BPs, for example to a foundation).

\subsection{Further Examples of TFMs}\label{app:tfm_ex}

TFMs can be assembled from different inclusion, confirmation, payment, and distribution rules in many natural ways. In addition to the rules given above, we give more examples here. All of the following rules are defined with respect to a game
structure (a user set~$I$, a BP set~$J$, BP transaction
sets~$\txs_1,\ldots,\txs_m$, and a block size~$k$). In all cases, ties are broken according to some consistent tie-breaking rule.

\begin{example}[Serial Dictatorship Inclusion Rule]
This inclusion rule is defined, for
every user bid vector $\bids$, by $\incs(\bids) =
(\block_1,\ldots,\block_m)$, where~$\blockj$ is chosen to maximize the sum
of the bids of the included transactions, subject to disjointness
with~$\block_1,\ldots,\block_{j-1}$ (and feasibility). That is,
$\block_j$ is the~$k$ highest-bidding transactions in $S_j \setminus
\cup_{h=1}^{j-1} \block_h$. (Or, if there are less than~$k$ such
transactions, all of them are included.) 
\end{example}

\begin{example}[Second-Price Auction (SPA) Confirmation Rule]
This confirmation rule confirms all but the lowest-bidding included
transaction. That is, $\conf(\blocks) = T(\blocks) \sm \{t\}$,
where~$t$ is the transaction of~$T(\blocks)$ with the lowest bid.
\end{example}

\begin{example}[Second-Price Auction (SPA) Payment Rule]
This payment rule charges~0 to the lowest-bidding included
transaction~$t$, and~$\bid_t$ to the other included transactions.
That is, $\pricei(\blocks) = \bid_t$ if $i \in T(\blocks) \sm
\{t\}$ and $\pricei(\blocks) = 0$ otherwise.
\end{example}

\begin{example}[The Null Distribution Rule]
This distribution rule burns all transaction fees: $\revj(\blocks) =
0$ for all~$\blocks$ and~$j$.
\end{example}

\begin{example}[Shapley Distribution Rule (FPA Version)]
This distribution rule splits the bid of each included transaction
equally among the BPs that included it. 
That is, 
\[
\revj(\blocks) = \sum_{i \in \blockj} \frac{\bidi}{m_i(\blocks)},
\]
where~$m_i(\blocks) = |\{h \in J \,:\, i \in B_h\}$ denotes the number of BPs
that included~$i$ in their block.
\end{example}

The distribution rule above is intended for use with the FPA payment
rule. The SPA version of the Shapley distribution rule is defined
similarly, except with~$\bidi$ replaced by the lowest bid of an
included transaction and with no BP earning any revenue from the
lowest-bidding transaction.

\begin{example}[Serial Dictatorship Distribution Rule (FPA Version)]
This distribution rule passes on revenue earned from an included
transaction to the lexicographically first BP that included it.
That is, 
\[
\revj(\blocks) = \sum_{i \in \blockj \sm \cup_{h=1}^{j-1} \block_h} \bidi.
\]
\end{example}

The distribution rule above is intended for use with the FPA payment
rule. The SPA version of the rule is defined
similarly, except with~$\bidi$ replaced by the lowest bid of an
included transaction and with no BP earning any revenue from the
lowest-bidding transaction. 

There are numerous ways to combine these rules or the rules described
in Section~\ref{ss:tfm_def} to produce natural TFMs. The FPA-EQ TFM is
analyzed at length in Section~\ref{ss:ub}. Other examples include:
\begin{enumerate}

\item {\em SPA-EQ}: WM inclusion rule, SPA confirmation rule, SPA payment
  rule, equal-share distribution rule (SPA version). (The SPA version
  of the equal-share distribution rule in~\eqref{eq:eq}
  replaces $\bidi$ by the lowest bid of an included transaction and
  sums only over the transactions of~$T(\blocks)$ other than the
  lowest-bidding one.)

\item {\em FPA-Shapley}: WM inclusion rule, FPA confirmation rule, FPA payment
  rule, Shapley distribution rule (FPA version).

\item {\em SPA-Shapley}: WM inclusion rule, SPA confirmation rule, SPA payment
  rule, Shapley distribution rule (SPA version).

\item {\em FPA-Serial}: serial dictatorship inclusion rule, FPA confirmation
  rule, FPA payment rule, serial dictatorship distribution rule (FPA
  version).

\item {\em SPA-Serial}: serial dictatorship inclusion rule, SPA confirmation
  rule, SPA payment rule, serial dictatorship distribution rule (SPA
  version).

\end{enumerate}

\subsection{Incentive Compatibility}\label{ss:ic}

Intuitively, a mechanism is incentive-compatible if its participants
are motivated to behave in a prescribed way, such as by bidding
truthfully (in the case of users) or by choosing blocks as instructed
by a TFM's inclusion rule (in the case of BPs).
We next formalize these two incentive-compatibility properties (one
for users, one for BPs).

\vspace{-.5\baselineskip}
\paragraph{Dominant-strategy incentive-compatibility (DSIC).}
We first observe that the composition of an (intended) inclusion
rule~$\incs$ and confirmation rule~$\conf$ of a TFM induce an
(intended) {\em allocation rule}~$\allocs$, with $\alloci(\bids) = 1$
if $i \in \conf(\incs(\bids))$ and $\alloci(\bids) = 0$ otherwise.
That is, $\allocs(\bids)$ is the characteristic vector of the
confirmed transactions with user bids $\bids$, assuming that the BPs
carry out the intended inclusion rule.
Under the same assumption, the payoff of user~$i$ under bid vector
$\bids$ in the TFM $\tfm$ is
\begin{equation}\label{eq:util}
\utili(\bids) = \vali \cdot \alloci(\bids) - \pricei(\incs(\bids)).
\end{equation}
A TFM is then {\em dominant-strategy incentive-compatible (DSIC)}
if, for every user~$i$, valuation~$\vali$, and bid vector~$\bids$,
$\utili(\vali,\bidsmi) \ge \utili(\bids)$. 
That is, 
after fixing the BP strategies to be those recommended by
the TFM's inclusion rule, truthful bidding is a dominant
strategy for every user. For example, in the BP-symmetric setting
(with $\txsj=I$ for all~$j \in J$), the SPA-EQ and SPA-Shapley TFMs
(see Section~\ref{app:tfm_ex}) are DSIC. TFMs that use the FPA
payment rule are never DSIC, as users are incentivized to shade their
bids.

\vspace{-.5\baselineskip}
\paragraph{Block producer incentive-compatibility (BPIC).}
In an outcome of a TFM $\tfm$, specified by the bids $\bids$ chosen by
users and the feasible allocation $\blocks$ chosen by BPs, the payoff
of BP~$j$ is $\revj(\blocks)$.
A TFM is then {\em block producer incentive-compatible (BPIC)} if, for
every bid vector $\bids$ with corresponding intended allocation
$\incs(\bids) = \blocks = (\block_1,\ldots,\block_m)$, every BP~$j$, and every
block~$\blockj'$ feasible for~$j$, $\revj(\blocks) \ge \revj(\blockj',\blocksmj)$.
That is, after fixing the user bids to~$\bids$, the feasible
allocation recommended by the TFM's inclusion rule is a Nash
equilibrium among the BPs.

For example, the SPA-EQ and SPA-Shapley TFMs from
Section~\ref{app:tfm_ex} are not BPIC,
as BPs generally have an incentive to deviate from the WM allocation
rule by including their own transactions in order to boost their
overall revenue. The FPA-Shapley TFM (see Section~\ref{app:tfm_ex})
fails to satisfy BPIC for a different reason:
BPs are generally incentivized to redundantly include a
  high-bid transaction multiple times rather that following the WM
  allocation rule (in which each transaction is included at most once). 

\vspace{-.5\baselineskip}
\paragraph{Strong BPIC.}
Despite the fact that many natural TFMs fail to satisfy it, the BPIC
condition is relatively weak. For example, any TFM that uses the null
distribution rule (with all transaction fees burned) is trivially
BPIC, with all BPs indifferent across all outcomes.
Thus, the BPIC condition does not generally provide much force toward BPs
carrying out the intended inclusion rule.

The next condition, a strengthening of BPIC, states that the intended
allocation should not merely be a Nash equilibrium, but should also be
strictly superior to all non-equivalent Nash equilibria. Formally, a
TFM $\tfm$ is {\em strongly BPIC} if, for every user bid vector
$\bids$, the following conditions hold:
\begin{enumerate}

\item the recommended feasible allocation $\blocks = \incs(\bids)$ is
  a Nash equilibrium among the BPs (holding user bids fixed at $\bids$);

\item every Nash equilibrium~$\blocks'$ among the BPs (again, with
  fixed bids~$\bids$) is either equivalent to or Pareto dominated by
  $\blocks$.

\end{enumerate}
Intuitively, two feasible allocations are ``equivalent'' if they are
the same up to tie-breaking and the inclusion of zero-bid
transactions. Formally, for a TFM $\tfm$, feasible allocations
$\blocks$ and $\blocks'$ are {\em equivalent} if the multi-sets of the
positive bids of the confirmed transactions $\conf(\blocks)$ and
$\conf(\blocks')$ are identical. We say that one allocation $\blocks$
{\em Pareto dominates} another allocation $\blocks'$ if: (i)
$\revj(\blocks) \ge \revj(\blocks')$ for all $j \in J$; and (ii)
$\revj(\blocks) > \revj(\blocks')$ for some $j \in J$.  We'll see in
Section~\ref{ss:fpaeq} an example of a strongly BPIC TFM (the FPA-EQ TFM).

\subsection{Approximate Welfare Guarantees}\label{ss:poa}

We assess the outcome quality of different TFMs using the
welfare objective~$W(\cdot)$, defined as the total value of the
confirmed transactions. That is, for a TFM $\tfm$ and feasible allocation
$\blocks$, $W(\blocks) = \sum_{i \in \conf(\blocks)} \vali$.
TFMs can suffer from welfare loss for three distinct
reasons. First, even if all participants behave as desired, a TFM's
inclusion rule may result in a suboptimal feasible allocation.
Second, even with the WM allocation rule and truthful bids, BPs may
coordinate on a suboptimal Nash equilibrium.
Third, even with the WM allocation rule and BPs that coordinate on the
intended Nash equilibrium, non-truthful bidding by users can lead to
suboptimal allocations.

Examples of welfare losses from the inclusion rule resulting in suboptimal allocations include:

\begin{itemize}

    \item In the BP-symmetric setting and with known valuations, and a FPA-Shapley TFM which we redefine to be BPIC. This would mean replacing the WM inclusion rule with an inclusion rule whereby the BPs maximize their personal revenue. Then, let $k=1$ be the block size, there is one transaction with bid~$b_1=m+\eps$, where $\eps\rightarrow 0$, and $m-1$ transaction with bids~$b_i =1$ for $i \in [2,\dots, m]$. Then all BPs would all include the first transaction, resulting in a welfare $\approx 2$ worse than optimal.

    \item Consider the following setting,  known user valuations, and a TFM serial dictatorship inclusion rule. The block size is $k=1$, there are $m=2$ BPs, and two transactions with bids $b_1=b_2=1$ of which the first BP can include both, while the second BP can only include the first transaction. The first BP would include the first transaction and there would be no transactions for the second BP to include -- the resulting welfare is a factor of $2$ worse than optimal.

\end{itemize}

For examples of welfare losses from BPs coordinating on a suboptimal Nash equilibrium even with the WM allocation rule and truthful bids consider the following:

\begin{itemize}
    \item Consider a FPA-EQ TFM in the BP-asymmetric setting. Again, the block size is $k=1$, there are $m=2$ BPs, and two transactions with bids $b_1=b_2=1$ of which the first BP can include both, while the second BP can only include the first transaction. The first BP including the first transaction and the second BP including no transaction is a Nash equilibrium that again results in a welfare $\approx 2$ worse than optimal.
\end{itemize}

Non-truthful bidding in stage 1 by users can lead to suboptimal allocations even with the WM allocation rule and BPs that coordinate on the as demonstrated by Vickrey~\cite{vickrey1961counterspeculation} by showing that the equilibria for a FPA (i.e., $k=m=1$)  are not generally efficient.

Thus, 
a equilibrium welfare approximation guarantee is a
guarantee that the welfare loss {\em from all three of these sources
  combined} is relatively modest.

\section{FPA-EQ: A Strongly BPIC and Near-Optimal TFM}\label{s:main}

\subsection{What Can We Hope For?}

We have highlighted three desirable properties of TFMs (in addition to
our standing requirements that TFMs be deterministic and
ex post individually rational): (i) DSIC; (ii)
strong BPIC; and (iii) optimal or near-optimal welfare at equilibrium.
In this work, we take the strong BPIC condition~(ii) as a hard
constraint.
(If BPs are not properly motivated to carry out the intended
inclusion rule, which in turn determines the confirmed transactions
and their payments, it's unclear how to interpret a proposed TFM.) However, insiting on DSIC and strong BPIC simultaneously leaves us only with trivial welfare approximations.  

\begin{theorem}\label{t:lb1}
    Any DSIC and strongly-BPIC TFM has a worst case welfare approximation of 0. 
\end{theorem}

For a TFM to be DSIC it must be that an agent's payment is a function of other agents' bids. However, this then gives the BP an opportunity to profit by manipulating a winning agent's payment via shill bidding. Strong BPIC also dictates that the BP must get positive revenue at equilibrium or else they could simply propose the empty block, so the TFM can't simply burn the winning agents' payments. It follows that any DSIC and strongly BPIC TFM must have a threshold value for inclusion upon which the BP gets paid a fixed amount independent of the other agents' bids. Hence, in an instance where agents' values all fall below these thresholds, the TFM outputs an empty block giving a worst case welfare approximation of 0. For a full proof  see Appendix \ref{app:lb1}. 

Theorem \ref{t:lb1} implies we have no choice but to consider non-DSIC TFMs. However, even without DSIC or strong BPIC as a constraint, no TFM can hope to always achieve optimal welfare at equilibrium. 

\begin{theorem}\label{t:lb2}
For any TFM, there exists a game structure and a
valuation distribution for which there is a Bayes-Nash equilibrium
with expected welfare strictly less than the minimum possible.
\end{theorem}

For a TFM to always achieve optimal welfare at equilibrium, it must be the case that the blockspace is always filled with the highest value transactions regardless of the realizations of bidders' values. Even in the simple case of a single BP that has a blocksize of 1 with two bidders, this is impossible to satisfy for any prior free mechanism. The issue is that a blocksize of 1 implies that the payment rule of the TFM must be a function only of the winner's bid. We show that given such a payment rule, bidders will always want to shade their bids, leading to the lower value bidder occasionally winning out. For details see Appendix \ref{app:lb2}.

In light of these negative results, the best-case
scenario is a strongly BPIC TFM that guarantees near-optimal welfare
at equilibrium. We present such a TFM next.

\subsection{The FPA-EQ TFM}\label{ss:fpaeq}

The rest of this section analyzes the {\em first-price auction with
  equal sharing (FPA-EQ)} TFM. The ingredients of this TFM were all
introduced in Section~\ref{ss:tfm_def}:
\begin{itemize}

\item the welfare-maximizing (WM) inclusion rule (i.e.,
  with~$\incs(\bids)=\blocks$ chosen to maximize the sum of the bids
  $\sum_{i \in T(\blocks)} \bidi$ of the included transactions, with
  ties broken according to some   consistent rule);

\item the FPA confirmation rule (with all included transaction
  confirmed: $\conf(\blocks) = T(\blocks)$);

\item the FPA payment rule (with each user of a confirmed transaction
  paying its bid);

\item the equal share (FPA version) distribution rule (with the
  payment for each confirmed transaction split equally between the $m$
  block producers, as in~\eqref{eq:eq}).

\end{itemize}

Because of its FPA payment rule, the FPA-EQ TFM is not DSIC; bidders
are incentivized to shade their bids. Unlike many other natural TFMs,
however, the FPA-EQ TFM is strongly BPIC. The proof of this fact leans
heavily on the choice of the equal-share distribution rule, and also
on the matroid structure of feasible allocations.

\begin{prop}[FPA-EQ Is Strongly BPIC]\label{prop:strong}
For every game structure, the FPA-EQ TFM is strongly BPIC.
\end{prop}

\begin{proof}
Fix a game structure and a user bid vector $\bids$.
The payoff of every BP is proportional to the total amount paid by users
(due to the equal-share distribution rule), and therefore to the sum
of the bids of the confirmed transactions (due to the FPA payment
rule), and therefore to the sum of the bids of the included
transactions (due to the FPA confirmation rule). Because the WM
allocation rule instructs BPs to maximize the sum of the bids of the
included transactions over feasible allocations, the intended
allocation~$\blocks^*$ is a Nash equilibrium among the BPs (holding
user bids fixed at $\bids$).
By the same reasoning, $\blocks^*$ Pareto dominates every Nash
equilibrium allocation that fails to maximize the sum of the bids of
the included transactions.
Finally, because the subsets of transactions that can be included in a feasible
allocation form a matroid (see Proposition~\ref{prop:matroid}) and
due to the lexicographic optimality property of matroids (see
Proposition~\ref{prop:lex}), every feasible allocation~$\blocks$ that
maximizes the sum of the included bids is equivalent to~$\blocks^*$
(i.e., after ignoring zero-bid transactions, the multi-sets of bids of
transactions in $\conf(\blocks)$ and $\conf(\blocks^*)$ are identical).
\end{proof}

\subsection{An Approximate Welfare Guarantee for FPA-EQ}\label{ss:ub}

Our main result in this section is that the FPA-EQ TFM, in addition to
satisfying the strong BPIC property (Proposition~\ref{prop:strong}),
achieves near-optimal welfare at equilibrium.  Precisely, in the
extensive-form game induced by this TFM $\tfm$ (see
Section~\ref{ss:game}), call a strategy profile {\em inclusion-rule
  respecting (IRR) at $\bids$} if, in the subgame corresponding to
$\bids$, the BPs choose a feasible allocation that is equivalent to
$\incs(\bids)$. (As in Section~\ref{ss:ic}, two feasible allocations
are equivalent if the resulting sets of confirmed transactions share
the same multi-sets of positive bids.) A subgame-perfect equilibrium
is then called inclusion-rule respecting if it is IRR at every user
bid vector $\bids$. For a strongly BPIC TFM like FPA-EQ, there is good
reason to focus on its IRR SPE---in any other SPE, there are bids
vectors for which BPs inexplicably coordinate on a subgame equilibrium
that is Pareto dominated by the one suggested by the TFM's inclusion
rule.

\begin{theorem}[FPA-EQ Is Approximately Welfare-Optimal]\label{t:ub}
For every game structure and valuation distribution~$\dists$, every
inclusion-rule-respecting subgame perfect equilibrium of the
FPA-EQ TFM has expected welfare at least $1-\tfrac{1}{e} \approx
63.2\%$ of the maximum possible.
\end{theorem}

The proof of Theorem~\ref{t:ub} proceeds in two steps. The first
step establishes an equivalence between the IRR SPE of the FPA-EQ TFM
and the Bayes-Nash equilibria of a (single-shot) winner-pays-bid
matroid auction. Intuitively, with the BP behavior fixed (up to
allocation equivalence) in an IRR SPE, we can analyze users as if they
are competing in a single-stage game. The second step of the proof applies the theory of smooth games (see
e.g.~\cite{RST17}) to prove a worst-case bound on the expected welfare
of the Bayes-Nash equilibria of winner-pays-bid matroid
auctions.\footnote{Such a bound was proved in~\cite{HHT15} for the
  special case of independent user valuations; the bound here for
  correlated user valuations appears to be new.}  

\paragraph{Equivalence of IRR SPE with BNE of Matroid Auctions.}
We first show a correspondence between the IRR SPE of the FPA-EQ TFM and the
Bayes-Nash equilibria of winner-pays-bid matroid auctions. Here's what
we mean by the latter:
For a set of users~$U$ and a matroid~$(U,\I)$ (see
Definition~\ref{d:matroid}), the corresponding winner-pays-bid matroid
auction is defined by:
\begin{enumerate}

\item Simultaneously, each user~$i \in U$ submits a nonnegative bid
  $\bidi$.

\item The mechanism chooses an independent set~$A \in \I$ that
  maximizes the sum~$\sum_{i \in A} \bidi$ of the bids of the included
  users, breaking ties arbitrarily. Users of~$A$ win and the other
  users lose.

\item Each winner~$i \in A$ pays its bid $\bidi$.

\end{enumerate}
For example, a first-price single-item auction corresponds to the
special case of a winner-pays-bid matroid auction in which the set~$\I$
contains only the empty set and all the singleton sets.

Every strategy of a user~$i \in I$ in the extensive-form game induced
by a TFM (see Section~\ref{ss:game}) induces a bidding
strategy~$\sigma_i$, with~$\sigma_i(\vali)$ defined as the action (or
distribution over actions) taken by user~$i$ in the information set
corresponding to the realization~$\vali$ of its valuation.  Meanwhile,
every profile of BP strategies induces an allocation rule~$\allocs$,
where~$\alloci(\bids)$ denotes the probability (over any randomness in
BPs' strategies) that user~$i$'s transaction is confirmed when the
user bid vector is~$\bids$.

For an arbitrary allocation rule~$\allocs$, the 
corresponding (single-shot) {\em winner-pays-bid
  mechanism}~$(\allocs,\prices)$ accepts nonnegative bids from users;
chooses a feasible allocation from a probability
  distribution such that each user~$i \in U$ is allocated with
  probability $\alloci(\bids)$; and charges $\bidi$ to each allocated
  user and~0 to each unallocated user.
Fron the discussion above, we have:
\begin{prop}\label{prop:equiv}
Every strategy profile in the extensive-form game induced by the
FPA-EQ TFM is user-outcome-equivalent to the induced bidding
strategies~$\sigma_1(\val_1),\ldots,\sigma_n(\val_n)$ in the
winner-pays-bid mechanism induced by the allocation rule that is
induced by BP's strategies.
\end{prop}
By ``user-outcome-equivalent'' we mean that, for each user, the
probability of allocation and the payment conditional on allocation
are identical in the two scenarios.
Note that this notion of equivalence preserves the expected welfare.

We now specialize Proposition~\ref{prop:equiv} to the case of IRR
SPE. First, the IRR condition means that the allocation rule $\allocs$
induced by the BP strategies is the one that, given users' bids, selects
the feasible allocation with the maximum-possible sum of bids
(breaking ties arbitrarily). Thus, the winner-pays-bid mechanism
induced by an IRR SPE is a matroid auction. Second, the equilibrium
condition for users' strategies in the IRR SPE translate to the
Bayes-Nash equilibrium conditions for the induced bidding
strategies~$\sigma_1(\val_1),\ldots,\sigma_n(\val_n)$ in this matroid
auction.

\begin{lemma}\label{l:equiv}
For every game structure and valaution distribution, every IRR SPE of
the FPA-EQ TGM is user-outcome-equivalent to a Bayes-Nash equilibrium
of a winners-pay-bid matroid auction (with the same valuation
distribution).
\end{lemma}
As noted above, user-outcome-equivalence implies that the expected
welfare of an IRR SPE and the corresponding Bayes-Nash equilibrium are
the same.

\paragraph{The Price of Anarchy of Winner-Pays-Bid Matroid Auctions.}
Given the equivalence established above, we can complete the proof of
Theorem~\ref{t:ub} by showing the following:

\begin{theorem}[Matroid Auctions Have Only Near-Optimal Equilibria]\label{t:poa}
For every matroid~$(U,\I)$ and valuation distribution, every
Bayes-Nash equilibrium of the corresponding winner-pays-bid matroid
auction has expected welfare at least $1-\tfrac{1}{e}$ times the
expected maximum welfare.
\end{theorem}

In turn, proving Theorem~\ref{t:poa} reduces to showing that
winner-pays-bid matroid auctions are ``smooth'' in a suitable sense.
The following definition and theorem are essentially due to Lucier and
Paes Leme~\cite{LP11}; we follow the formalism in Roughgarden et
al.~\cite[Definition 4.5; Theorem 4.6]{RST17}.

\begin{definition}[Smooth Auction with Private
  Deviations~\cite{LP11,RST17}]\label{d:smooth}
For parameters $\lambda \ge 0$ and $\mu \ge 1$,
an auction with allocation rule~$\allocs$ and payment rule $\prices$ is
  {\em $(\lambda,\mu)$-smooth with private deviations} if for every
  valuation profile $\vals$
there exist probability distributions
  $D^*_1(\val_1),\ldots,D^*_n(\val_n)$ over bids
  such that, for every bid profile $\bids$,
\begin{equation}\label{eq:smooth}
\sum_i \expect[\bidi^*\sim D_i^*(\val_i)]{u_i(\bidi^*,\bids_{-i})} \geq \lambda \cdot \sum_i v_i \cdot \alloc^*_i(\vals) - \mu \cdot Rev(\bids).
\end{equation}
\end{definition}
In~\eqref{eq:smooth},
$u_i(\bids) = v_i \cdot \alloci(\bids) - \pricei(\bids)$ denotes
quasi-linear utility (as in~\eqref{eq:util}), $\alloc^*(\vals)$
denotes the characteristic vector of a welfare-maximizing feasible
solution with respect to valuation profile~$\vals$, and
$Rev(\bids) = \sum_i \pricei(\bids)$ denotes the auction's revenue
when the bid vector is $\bids$.  The ``private deviations'' qualifier
refers to the fact that each bid distribution~$D^*_i$ is permitted to
depend only on user~$i$'s valuation $\vali$, and not on the full
valuation profile~$\vals$.

\begin{theorem}[Smoothness Implies Price-of-Anarchy
  Bounds~\cite{LP11,RST17}]\label{t:smooth}
If an auction is $(\lambda,\mu)$-smooth,
then for every distribution $\dists$ over players' valuations,
every Bayes-Nash equilibrium of the
auction has expected welfare at least $\lambda/\mu$ times the expected
maximum welfare.
\end{theorem}

In light of Theorem~\ref{t:smooth}, the following lemma implies
Theorem~\ref{t:poa} (and hence, by Lemma~\ref{l:equiv},
Theorem~\ref{t:ub}).

\begin{lemma}[Matroid Auctions Are Smooth]\label{l:smooth}
For every matroid~$(U,\I)$, the corresponding winner-pays-bid matroid
auction is~$(1-\tfrac{1}{e},1)$-smooth with private deviations.
\end{lemma}

\begin{proof}
The proof incorporates elements of the smoothness analysis of first-price
auctions by Syrgkanis and Tardos~\cite{ST13} and the revenue covering
analysis of matroid auctions by Hartline et al.~\cite{HHT15}.
Fix a matroid~$(U,\I)$; let $\allocs$ and $\prices$ denote the
allocation and payment rules of the corresponding winner-pays-bid
auction.
Fix a valuation profile~$\vals$ for the users of~$U$. 
For each~$i \in U$, define $D^*_i$ as the distribution with density
$1/(\vali-x)$ on support $[0,(1-1/e)\vali]$.

To verify the smoothness inequality~\eqref{eq:smooth}, fix a bid
vector $\bids$.  Denote by~$t_i(\bidsmi)$ the minimum value~$z$
for~$i$'s bid such that $\alloci(z,\bidsmi) = 1$.  To bound
$\expect[\bidi^*\sim D_i^*]{u_i(\bidi^*,\bids_{-i})}$, we consider two
cases.  First, if $\vali \cdot (1-1/e) \le t_i(\bidsmi)$, then because
$\bidi^*\sim D_i^*$ is at most $\vali$ with probability~1
and~$(\allocs,\prices)$ is ex post individually rational,
$\expect[\bidi^*\sim D_i^*]{u_i(\bidi^*,\bids_{-i})} \ge 0$.  Second,
if $\vali \cdot (1-1/e) > t_i(\bidsmi)$, then by similar reasoning,
\[
\expect[\bidi^*\sim D_i^*]{u_i(\bidi^*,\bids_{-i})} \ge
\int_{t_i(\bidsmi)}^{(1-1/e)\vali} (\vali - z) \cdot \frac{dz}{\vali -
  z} = \left(1 - \frac{1}{e} \right) \vali - t_i(\bidsmi).
\]
In this case, because the left-hand side is nonnegative and
$\alloci(\vals) \in [0,1]$, we also have
\[
\expect[\bidi^*\sim D_i^*]{u_i(\bidi^*,\bids_{-i})} \ge
\left(1 - \frac{1}{e} \right) \vali \cdot \alloci^*(\vals) -
t_i(\bidsmi) \cdot \alloci^*(\vals).
\]

Summing this inequality over all~$i \in I$ and applying
Proposition~\ref{prop:rc}, we have
\begin{align*}
\expect[\bidi^*\sim D_i^*]{u_i(\bidi^*,\bids_{-i})} 
& \ge \left(1 - \frac{1}{e} \right) \sum_{i \in I} \vali \cdot
  \alloci^*(\vals) - \sum_{i \in I} t_i(\bidsmi) \alloci^*(\vals)\\
& \ge \left(1 - \frac{1}{e} \right) \sum_{i \in I} \vali \cdot
  \alloci^*(\vals) - \sum_{i \in I} \pricei(\bids),
\end{align*}
which shows that~\eqref{eq:smooth} holds with~$\lambda =
1-\tfrac{1}{e}$ and $\mu = 1$, completing the proof.
\end{proof}

We can obtain stronger guarantees if we impose symmetry conditions on the
BPs and users. In the BP-symmetric setting (see
Section~\ref{ss:players}), a simple exchange argument shows that {\em
  every} SPE of the FPA-EQ TFM is IRR. Thus:

\begin{corollary}\label{c:symmetric}
In the BP-symmetric setting, 
for every game structure and valuation distribution~$\dists$, every
subgame perfect equilibrium of the
FPA-EQ TFM has expected welfare at least $1-\tfrac{1}{e} \approx
63.2\%$ of the maximum possible.
\end{corollary}

Adapting an example of Syrgkanis~\cite{S14} for first-price auctions
to the present setting gives a lower bound showing that the approximation factor
of~$1-\tfrac{1}{e}$ in Theorem~\ref{t:ub} and Corollary~\ref{c:symmetric} is tight. 

\begin{prop}[Theorem~\ref{t:ub} Is Tight]\label{prop:s14}
There exists a game structure, valuation distribution~$\dists$, and
an inclusion-rule-respecting subgame perfect equilibrium of the FPA-EQ
TFM with expected welfare~$1-\tfrac{1}{e}$ times the expected maximum
welfare.
\end{prop}

\begin{proof}
Take~$I=\{1,2,3\}$, $J=\{1\}$, and $\txs_1 = \{1,2,3\}$.
The support of the joint distribution~$\dists$ is the valuation vectors of
the form~$(1,x,x)$ for~$x \in [0,1-\tfrac{1}{e}]$. The marginal
distribution of the common value of~$\val_2$ and~$\val_3$ is given by
the CDF~$F(x) = \tfrac{1}{e} \tfrac{1}{1-x}$ on $[0,1-\tfrac{1}{e}]$.
Thus, with probability~1, the maximum-possible welfare is~1 (achieved
by including the first transaction).

One can check that the following is an IRR SPE. The BP includes the
highest-bidding transaction, breaking ties in favor of the first
transaction. The first user always bids~0. The second and third users
always bid truthfully. A calculation shows that the expected welfare
of this IRR SPE is exactly~$1-\tfrac{1}{e}$.
\end{proof}

If we further assume that users are symmetric, meaning that their
valuations are drawn i.i.d.\ from a common distribution, then every
SPE of the FPA-EQ TFM is in fact fully efficient. The following
corollary follows from Lemma~\ref{l:equiv}
and the full efficiency of
Bayes-Nash equilibria in multi-unit auctions with symmetric
unit-demand bidders (see e.g.~\cite{krishna}):

\begin{corollary}[Optimal Welfare in Symmetric Settings]
In the BP-symmetric setting, 
for every game structure and i.i.d.\ valuation distribution, every
subgame perfect equilibrium of the
FPA-EQ TFM achieves the maximum-possible expected welfare.
\end{corollary}

As noted in Section~\ref{ss:game}, these positive results assume
  that BPs are capable of coordinating on an equilibrium of the
  appropriate type. It would be interesting to investigate how our
  guarantees would change under weaker versions of this assumption.

\bibliographystyle{splncs04}
\bibliography{dagtfm}

\appendix

\section{Supplementary Material for Section~\ref{s:main}}

\subsection{Review of Relevant Matroid Theory}

The matroid structure of feasible allocations play an important role
in the incentive-compatibility and welfare guarantees of the FPA-EQ
mechanism in Section~\ref{s:main}. We review in this appendix the
properties of matroids that are relevant to our results.

\begin{definition}[Matroid]\label{d:matroid}
A {\em matroid} is a set system~$(X,\I)$ with {\em ground set}~$X$ and
{\em independent sets}~$\I \sse 2^X$ that satisfies:
\begin{enumerate}

\item $\I$ is non-empty.
\item (Downward closure) If $A' \in \I$ and $A \sse A'$, then $A \in \I$.
  \item (Exchange property) If $A,A' \in \I$ with $|A'| > |A|$, then
    there exists $x \in A' \sm A$ such that $A \cup \{x\} \in \I$.
    \end{enumerate}  
\end{definition}  

For a game structure~$\gs$, call a subset $A \sse I$ of transactions
{\em feasible} if there exists a feasible
allocation~$(B_1,\ldots,B_m)$ that includes precisely the transactions
in~$A$.

\begin{prop}\label{prop:matroid}
For every game structure~$\gs$, the subset of feasible
transactions forms a matroid over~$I$.
\end{prop}

\begin{proof} (Sketch.)
  Non-emptyness holds because the empty set of transactions is feasible.
  Downward closure holds because removing transactions from a feasible
  allocation cannot destroy feasibility. The exchange property holds
  from an alternating path argument in the spirit of transversal
  matroids (see~\cite[Theorem 1.6.2]{oxley}).
\end{proof}

Matroids have a long list of nice properties.

\begin{prop}[Lexicographic Optimality]\label{prop:lex}
Let~$(X,\I)$ be a matroid for which each ground set element~$x \in X$
has a nonnegative weight~$w_x$. If~$A,A' \in \I$ are two
maximum-weight independent sets, then the multi-sets of non-zero
element weights of~$A$ and~$A'$ are identical.
\end{prop}

\begin{proof}
(Sketch.) We can assume that $A,A'$ are maximal independent sets,
extending them with (necessarily zero-weight) elements if necessary.
Due to the matroid structure (see~\cite[Corollary 1.2.5]{oxley}),
there is a sequence~$A=A_0,A_1,\ldots,A_l=A'$ such that: (i) each set
in the sequence belongs to~$\I$; and (ii) each set in the sequence 
is derived from the previous one by swapping one element for
another. Because both~$A$ and $A'$ are maximum-weight independent
sets, so are all the intermediate sets of the sequence. Thus, each
swap of the sequence exchanges one element for another with equal
weight. Thus, the multi-sets of element weights of~$A$ and~$A'$ are
identical.
\end{proof}  

The following proposition establishes a ``revenue covering'' property
(in a sense similar to Hartline et al.~\cite{HHT15}) for matroids.

\begin{prop}[Revenue Covering]\label{prop:rc}
Let~$(X,\I)$ be a matroid for which each ground set element~$x \in X$
has a nonnegative weight~$w_x$, and let~$A^*$ denote a maximum-weight
independent set.
Let~$t_{x}(\mathbf{w}_{-x})$ denote
the minimum value of~$x$'s weight such that, holding the weights
$\mathbf{w}_{-x}$ of the other elements fixed, $x$ belongs to a
maximum-weight independent set.
Then, for every independent set~$A$,
\begin{equation}\label{eq:rc}
\sum_{x \in A^*} w_x \ge \sum_{x \in A} t_x(\mathbf{w}_{-x}).
\end{equation}
\end{prop}

\begin{proof}
(Sketch.)
By the
optimality of the greedy algorithm for matroids (see~\cite[Theorem
1.2.6]{oxley}), $A^*$ remains a
maximum-weight independent set even after the weight
of each element $x \notin A^*$ is increased to $t_x(\mathbf{w}_{-x})$.
Given that $w_x \ge t_i(\mathbf{w}_{-x})$ for
all $x \in A^*$ (by the definition of the~$t_i$'s), the
inequality~\eqref{eq:rc} follows from the optimality of~$A^*$.
\end{proof}

\subsection{Proof of Theorem \ref{t:lb1}}\label{app:lb1}

We show here that insisting on DSIC and strong BPIC implies that the
TFM must output the empty set for some bid vectors, precluding it from
getting any welfare guarantees. Since we are considering DSIC
mechanisms, we consider the welfare achieved when bidders bid
truthfully.

The theorem follows immediately from the following lemma.

\begin{lemma}
    For any DSIC and strongly-BPIC TFM, for all $l>0$ and $\txss$ , there exists a valuation vector $\bold{v}$ where $v_l>0 \ \ \forall i\in [l]$ and $v_i=0$ otherwise, such that $x_i(\bold{v},\txss) = 0 \ \forall i\in I$ i.e. the TFM confirms no transactions.
\end{lemma}

\begin{proof}

    We proceed by induction on $l$. For the base case of $l=1$, let $\bold{v}^1=(v_1,0,...,0)$ and consider an arbitrary $\txss$. Assume for the sake of contradiction that $x_1(\bold{v}^1,\txss)=1$ for all $v_1>0$. Since the TFM is DSIC, by Myerson's Theorem, we have that $p_1(\incs(\bold{v}^1)) = 0$. It follows that $ \pi_j(\conf(\incs(\bold{v}^1,\txss))) = 0$ for all $j\in J$ since no transactions make any non-zero payments. However, then we have that all the block producers are indifferent between the equilibria $B_j = \emptyset$ for all $j\in J$ and $B_j = y_j(\bold{v}^1,\txss)$  contradicting the TFM being strongly-BPIC since these equilibria confirm different sets of bids.  Hence for the TFM to be DSIC and strongly-BPIC there exists a $v_1 >0$ s.t. $x_i(\bold{v}^1,\txss) = 0 \ \forall i\in I$.

    For the inductive hypothesis assume for all $\txss$, there exists a valuation vector $\bold{v}^l=(v_1,...v_l,0,...,0)$ s.t.  $x_i(\bold{v}^l,\txss)=0 \ \forall i\in I$ . We then show for all $\txss'$ there exists a $v'>0$ s.t. for $\bold{v}^{l+1}$ with $\bold{v}_i^{l+1}=\bold{v}_i^{l}$ for $i \neq l+1$ and $\bold{v}_{l+1}^{l+1}=v'$, $x_i(\bold{v}^{l+1},\txss')=0 \ \forall i\in I$

    Given a $\txss'$ let $\bold{v}^l$ be a valuation vector such that $\allocs(\bold{v}^l,\txss) = \emptyset$ where $\txss$ is the projection of $\txss'$ to transactions $i\neq l+1$. Now assume for the sake of contradiction that $\allocs(\bold{v}^{l+1},\txss') \neq \emptyset$ for all $v'>0$. We claim this implies that we must have $l+1 \in \allocs(\bold{v}^{l+1},\txss') \ \forall v'>0$ . This is because if there is a $v'>0$ where $l+1 \notin \allocs(\bold{v}^{l+1},\txss')$, we have two cases, either $\sum_j \pi_j(\incs(\bold{v}^{l+1},\txss')) > 0$ or $\sum_j \pi_j(\incs(\bold{v}^{l+1},\txss')) = 0$. The former case would imply the TFM is not strongly-BPIC, since in the instance with valuation function $\bold{v}^l$ and $\txss$, the BPs could censor transaction $l+1$ and replace it with a transaction with bid $v'$. Then the BP's following $\incs$ under this modified valuation vector would pareto dominate following $\incs$ under $\bold{v}^l$ since $l+1$ isn't confirmed under $\bold{v}^{l+1}$, hence paying 0 fees, and some BPs get strictly positive compared to 0 revenue.  Otherwise when $\sum_j \pi_j(\incs(\bold{v}^{l+1},\txss')) = 0$, the BPs are indifferent between playing $\incs$ or all proposing $B_j=\emptyset$ also violating strong-BPIC.

    However, $l+1 \in \allocs(\bold{v}^{l+1},\txss') \ \forall v'>0$ implies that $p_{l+1}(\bold{v}^{l+1}) = 0$ by Myerson's Theorem. Now again we have the same two cases, either $\sum_j \pi_j(\incs(\bold{v}^{l+1},\txss')) > 0$ or $\sum_j \pi_j(\incs(\bold{v}^{l+1},\txss')) = 0$. Since we still have that $l+1$ can costlessly be included, the same reasoning applies contradicting the TFM being strongly-BPIC. Thus for any $\txss'$ there must exist a $v'>0$ such that  $x_i(\bold{v}^{l+1},\txss')=0 \ \forall i\in I$.
\end{proof}

\subsection{Proof of Theorem \ref{t:lb2}}\label{app:lb2}

We now show that no TFM can always be fully efficient at equilibrium when bidders draw their values from asymmetric distributions. We effectively reduce our setting to the case of auctioning a single item where the payment rule is forced to only be a function of the winning bidder's bid. We use revenue equivalence with a second price auction to show that efficient equilibrium can't be implemented with these types of payment rules.

Consider the case where $n=2$, $m=1$, and $k=1$. In this case, the TFM is equivalent to a single item auction with two bidders. Furthermore, $m=1$ implies that the TFM's payment rule can only be a function of the winning bid. Thus the theorem follows immediately from the following lemma. 

\begin{lemma}
For any mechanism where the payment rule is a function only of the winning bid, there exists a valuation distribution for which there is a Bayes--Nash equilibrium whose expected welfare is strictly less than the maximum possible.
\end{lemma}

\begin{proof}
Suppose, for the sake of contradiction, that there is a mechanism $(x,p)$ whose payment rule depends only on the winning bid, i.e.\ $p_i(\mathbf{b}) = f(b_i)$ for some function $f$, such that \emph{every} Bayes-Nash equilibrium in this mechanism is efficient. We will exhibit two different valuation instances and argue that the mechanism cannot have an efficient BNE in both instances simultaneously.

\medskip
\noindent
\textbf{Instance 1:}
Consider an instance \(\mathcal{I}_1\) with two bidders where 
\[
v_1 \sim \text{Uniform}\bigl([0,100]\bigr) 
\quad\text{and}\quad
v_2 \sim \text{Uniform}\bigl([0,1]\bigr).
\]
Let $\sigma(\cdot) = \bigl(\sigma_1(\cdot),\sigma_2(\cdot)\bigr)$ be a Bayes--Nash equilibrium under $(x,p)$ that implements the \emph{efficient outcome} for all realizations of \((v_1,v_2)\). Since \(\sigma\) implements the efficient outcome for all values, it matches the allocation of a second-price auction.  By  revenue equivalence, the bidders’ expected payments under \(\sigma\) in this mechanism must match their expected payments when bidding truthfully in a second-price auction.

\smallskip
\noindent

\begin{itemize}
\item For bidder 1:
Under the assumption that $\sigma$ is efficient, we have that the probability bidder~1 wins is $Pr[v_1>v_2]$. Thus $E_{v_2}[p_1\bigl(\sigma_1(v_1),\sigma_2(v_2)\bigr) \ |\  v_1] = Pr[v_1>v_2] \cdot f\bigl(\sigma_1(v_1)\bigr)$. On the other hand, bidder 1's expected payment in a second price auction is $Pr[v_1>v_2] \cdot E[ v_2 \ |\  v_2<v_1]$. Note that $E[v_2 \ | \ v_2<v_1] = \min\{\frac{1}{2},\frac{v_1}{2}\}$ giving us  \[f\bigl(\sigma_1(v_1)\bigr) = E\bigl[v_2 \mid v_2 < v_1\bigr] = \min\{\frac{1}{2},\frac{v_1}{2}\} \implies \sigma_1(v_1)\in f^{-1}\bigl(\min\{\frac{1}{2},\frac{v_1}{2}\}\bigr)\]

\smallskip
\noindent
\item For bidder 2:
By a symmetric argument, for bidder 2 we have
\[
f\bigl(\sigma_2(v_2)\bigr) 
\;=\;
E\bigl[v_1 \mid v_1 < v_2\bigr] = \frac{v_2}{2} \implies \sigma_2(v_2) \;\in\; f^{-1}\!\Bigl(\tfrac{v_2}{2}\Bigr)
\]

\end{itemize}

\medskip
\noindent
\textbf{Instance 2:}
Now consider another two-bidder instance $\mathcal{I}_2$, where
\[
v_1 \sim \text{Uniform}\bigl([0,\tfrac{3}{2}]\bigr)
\quad\text{and}\quad
v_2 \sim \text{Uniform}\bigl([0,100]\bigr).
\]
Let $\sigma'(\cdot) = \bigl(\sigma'_1(\cdot),\,\sigma'_2(\cdot)\bigr)$ be an efficient BNE for this second instance.  A parallel revenue-equivalence argument tells us that for bidder~1 and bidder 2 we have:
\[
f\bigl(\sigma'_1(v_1)\bigr) 
\;=\;
E\bigl[v_2 \mid v_2 < v_1\bigr]
\;=\;
\tfrac{v_1}{2} \implies \sigma'_1(v_1) \;\in\; f^{-1}\!\Bigl(\,\tfrac{v_1}{2}\Bigr),
\]  
\[
f\bigl(\sigma'_2(v_2)\bigr) 
\;=\;
E\bigl[v_1 \mid v_1 < v_2\bigr] = \min\!\Bigl\{\tfrac{3}{4},\,\tfrac{v_2}{2}\Bigr\} \implies \sigma'_2(v_2) 
\;\in\; f^{-1}\!\Bigl(\,\min\{\tfrac{3}{4},\tfrac{v_2}{2}\}\Bigr).
\]

\medskip
\noindent

We claim that it is impossible for both \(\sigma\) to be an efficient BNE in \(\mathcal{I}_1\) \emph{and} \(\sigma'\) to be an efficient BNE in \(\mathcal{I}_2\).  To see why, consider the following deviation arguments:

\begin{enumerate}
\item 
\emph{Deviation of bidder 2 in instance \(\mathcal{I}_1\).}
  Suppose in \(\mathcal{I}_1\) that bidder~2, whenever \(v_2 > \tfrac{3}{4}\), chooses a random sample \(a \sim \text{Uniform}([0,100])\) and then plays \(\sigma'_2(a)\) instead of \(\sigma_2(v_2)\). Call this strategy $\tilde{\sigma}_2(\cdot)$ If \(\sigma\) is indeed a BNE in \(\mathcal{I}_1\), this deviation cannot increase bidder~2’s expected utility for \emph{any} \(v_2\).

  Now consider the specific value \(v_2 = \tfrac{7}{8}\).  Under \(\sigma\), bidder~2’s expected utility is
  \[
  E_{v_1}\bigl[u_2\bigl(\sigma_1(v_1),\sigma_2(\tfrac{7}{8})\bigr)\bigr]
  \;=\;
  \Pr\bigl[v_1 < \tfrac{7}{8}\bigr]\times \frac{7}{16}
  \;\;<\; 
  \frac{1}{200},
  \]

  By deviating to the $\tilde{\sigma}_2$ by sampling $a$ and playing $\sigma'_2(a)$, bidder~2’s expected utility when $v_2=\frac{7}{8}$ is
  \[
    E_{v_1,a}\bigl[u_2\bigl(\sigma_1(v_1),\,\sigma'_2(a)\bigr)\bigr]
    \;=\;
    \Pr\bigl[x\bigl(\sigma_1(v_1),\,\sigma'_2(a)\bigr) = 2\bigr]
    \;\times\;
    \tfrac{1}{8}.
  \]
  For \(\sigma\) to remain an equilibrium, we must therefore have
  \[
    \Pr\bigl[x\bigl(\sigma_1(v_1),\,\sigma'_2(a)\bigr)=2\bigr]
    \;<\;
    \frac{1}{25}.
  \]
  
\item 
\emph{Deviation of bidder 1 in instance \(\mathcal{I}_2\).}
  Next, suppose in \(\mathcal{I}_2\) that bidder~1, whenever \(v_1 > \tfrac12\), samples \(b\sim\text{Uniform}([0,100])\) and plays \(\sigma_1(b)\) from instance~\(\mathcal{I}_1\).  From the probability bound above, whenever \(v_1 > \tfrac12\)  the probability that bidder~1 wins against \(\sigma'_2(v_2)\) is at least \(\tfrac{24}{25}\).  In particular, at \(v_1 = \tfrac{3}{4}\), bidder~1’s expected utility from this deviation is at least 
  \[
  \left(\frac{3}{4}-\frac{1}{2}\right)\times \frac{24}{25}
  \;=\;
  \frac{6}{25},
  \]
  which is substantially larger than the at most \(\tfrac{1}{100}\) expected utility bidder 1 achieves under playing $\sigma'_1$, contradicting $\sigma'$ being a Bayes--Nash equilibrium. 
\end{enumerate}

Since we have derived a profitable deviation in one instance assuming that the other instance has an efficient BNE, it follows that \(\sigma\) and \(\sigma'\) cannot both be equilibria of their respective instances under the same payment rule \(f\). Thus any mechanism where the payment rule is only a function of the winning bid must have an inefficient Bayes--Nash equilibrium either in instance $\mathcal{I}_1$ or $\mathcal{I}_2$. 
\end{proof}

\end{document}